\documentclass[conference,letterpaper]{IEEEtran}
\usepackage[utf8]{inputenc} 
\usepackage[T1]{fontenc}
\usepackage{url}
\usepackage{ifthen}
\usepackage{cite}
\usepackage[cmex10]{amsmath}
\usepackage{graphics} 
\usepackage{subfig}
\usepackage{epsfig}
\usepackage{amssymb}  

\usepackage{amsthm}
\usepackage{multirow}
\usepackage{textcomp}
\usepackage{color}
\usepackage{soul}
\usepackage{blindtext}
\usepackage{xpatch}
\usepackage{lastpage}
\usepackage{balance}
\usepackage{fancyhdr} 

\xpatchcmd{\maketitle}{\thispagestyle{empty}}{}{}{}

\newtheorem{theorem}{Theorem }

\newcounter{lemmaproof}
\newenvironment{lemmaproof}[1][]{\refstepcounter{lemmaproof} \noindent{} }{\hfill $\blacksquare$}

\newtheorem{lemma}{Lemma}

\newtheorem{claim}{Claim}

\newcommand{\icol}[1]{
  \begin{bmatrix}#1\end{bmatrix}%
}

\newcommand{\irow}[1]{
  \begin{bmatrix}#1\end{bmatrix}%
}

\newcommand{\In}{{\rm In}}
\newcommand{\Out}{{\rm Out}}
\newcommand{\rk}{{\rm rk}}

\newcommand{\mincut}{{\rm mincut}}
\newcommand{\I}{{\rm I}}
\newcommand{\Hi}{{\rm H}}

\begin{document}

\setstcolor{red}
\title{Secure Network Coding in the Setting in Which a Non-Source Node May Generate Random Keys} 


\author{%
  \IEEEauthorblockN{Debaditya Chaudhuri}
  \IEEEauthorblockA{University at Buffalo\\
                    Email: debadity@buffalo.edu}
  \and
  \IEEEauthorblockN{Michael Langberg}
  \IEEEauthorblockA{University at Buffalo\\
                    Email: mikel@buffalo.edu}
  \and
  \IEEEauthorblockN{Michelle Effros}
  \IEEEauthorblockA{California Institute of Technology\\
                    Email: effros@caltech.edu}
}
\maketitle

\begin{abstract}
It is common in the study of secure multicast network coding in the presence of an eavesdropper that has access to $z$ network links, to assume that the source node is the only node that generates random keys. 
In this setting, the secure multicast rate is well understood.
Computing the secure multicast rate, or even the secure unicast rate, in the more general setting in which all network nodes may generate (independent) random keys is known to be as difficult as computing the (non-secure) capacity of multiple-unicast network coding instances --- a well known open problem.
This work treats an intermediate model of secure unicast in which only one node can generate random keys, however that node need not be the source node. 
The secure communication rate for this setting is characterized again with an eavesdropper that has access to $z$ network links. 
\end{abstract}

\section{Introduction}

In this work, we study secure network communication over a directed acyclic network $\mathcal{G} = (\mathcal{V},\mathcal{E})$ having a single source node $S$, a single terminal node $T$, and a single node $K$,  which is capable of generating random ``keys'' independent of the messages generated by $S$. We employ a notion of secure ``wiretap'' communication networks introduced by Cai and Yeung in \cite{cai2002secure} and studied further in, for example \cite{feldman2004capacity,cai2007security,yeung2008optimality,el2012secure,silva2011universal}.
Under this notion of security, given a communication scheme over $\mathcal{G}$, we consider an edge $e \in \mathcal{E}$ of the network to be secure in the presence of a wiretap adversary if and only if $I(M;X_e) = 0$, where $M$ denotes the source message and $X_e$ denotes the information communicated on edge $e$.\footnote{Detailed definitions of all concepts discussed here and below appear in Section~\ref{MOD}.}  
To be secure in the presence of an adversary that wiretaps any size-$z$ subset $\mathcal{W} = \{e_1,\cdots,e_z\} \subset \mathcal{E}$ of edges, we require that $I(M;X_\mathcal{W}) = 0$, where $X_\mathcal{W} = (X_{e_1},\cdots,X_{e_z})$. 

Given integers $R$ and $z$, we define a secure network code over the network $\mathcal{G}$ to be $(R,z)$-feasible if it allows information to be communicated from the source $S$ to the terminal $T$ at rate $R$ and, in addition, it secures the network against a wiretap adversary that eavesdrops on up to $z$ edges of the network. Our work entails determining, for each $z$, the closure of the set of rates that are $(R,z)$-feasible, thereby deriving the capacity-security region. 

When $K=S$, the capacity-security region for secure multicast network codes is well understood \cite{cai2002secure, feldman2004capacity} with several follow up works \cite{cai2007security,yeung2008optimality,el2012secure,silva2011universal} that address various methods to alter any given non-secure linear network code into a new code that is secure. 
In contrast, determining the capacity-security region for secure network codes over a single-source single-terminal network, where every node can generate random keys, is as hard as the problem of characterizing the (non-secure) capacity region of the $k$-unicast problem as shown by \cite{huang2018}. Results of a similar nature are also presented in \cite{chan2014network}. The $k$-unicast problem is a well known open problem in the study of network codes \cite{chan2014network,6293890,langberg2009multiple,jalali2012capacity,4460828}.   

In this work, we seek to make progress in the apparently difficult generalization from the scenario where only the source can generate random keys to the scenario where all nodes can generate keys by studying the case where only a single node can generate keys but allowing that single node to be arbitrary. 
Our central result is a characterization of the capacity-security region in the unicast (single-source single-terminal) setting when only a single network node $K \ne S \in \mathcal{V}$ can generate random keys.



The remainder of the paper is organized as follows. In Section \ref{MOD}, we present our model and preliminary notation. Our main result, the capacity-security characterization of the networks at hand, appears in Section~\ref{RES}. 
The characterization is combinatorial in nature and  involves different cut-set bounds between the source node, the key generating node, and the terminal node. 
Achievability is proven in Section \ref{TH1_ACH} via a reduction from secure communication over $\mathcal{G}$ to (non-secure) multi-source multi-cast network coding over a modified network $\mathcal{G}^*$ as shown in Figure \ref{SDAGR}. 
The converse proof, which is based on cutset bounds, appears in Section \ref{TH1_CONVa_body}. An additional converse proof, in the more general context of cyclic networks, is presented in Appendix \ref{TH1_CONV}. The proofs of some one of our lemmas and claims are presented in Appendix \ref{CS_LEM} and Appendix \ref{CLAIMPROOFS}, respectively.

\section{Network Model}
\label{MOD}

Our system model consists of the following components:
\begin{itemize}
\item [(a)] A finite directed acyclic graph $\mathcal{G} = \{\mathcal{V},\mathcal{E}\}$. We assume that each edge $e \in \mathcal{E}$ noiselessly transmits one unit of information (i.e., one field element in a given field $\mathbb{F}_q$) per unit time. We use multiple edges to model an edge with the ability to communicate more than one information symbol per unit time.

\item [(b)] A source node $S$, which generates a source message vector of length $R$, $M = \irow{M_1 & M_2 & \cdots & M_{R}}^T$, with $M_1, M_2, \cdots, M_{R}$ independently and uniformly distributed over the field $\mathbb{F}_q$ of size $q$.  

\item [(c)] A terminal node $T \in \mathcal{V}$, which is required to decode all the messages generated by the source $S$ with zero error.

\item [(d)] A node $K \in \mathcal{V}$, which generates a random ``key'' vector, $N = \irow{N_1, \cdots, N_{|N|}}^T$ with $N_1, \cdots, N_{|N|}$ independently and uniformly distributed over the field $\mathbb{F}_q$ with $N$ independent of $M$.  

\item [(e)] An eavesdropper that can access any subset $\mathcal{W} \subset \mathcal{E}$ of edges for which $|\mathcal{W}| \leq z$.
\end{itemize}

In the following subsections, we introduce our definition of a network code and discuss the notions of topological order and cut sets.

\begin{figure}[!t]
\centering
\subfloat[]{\includegraphics[width=1.55in, height=1.1in]{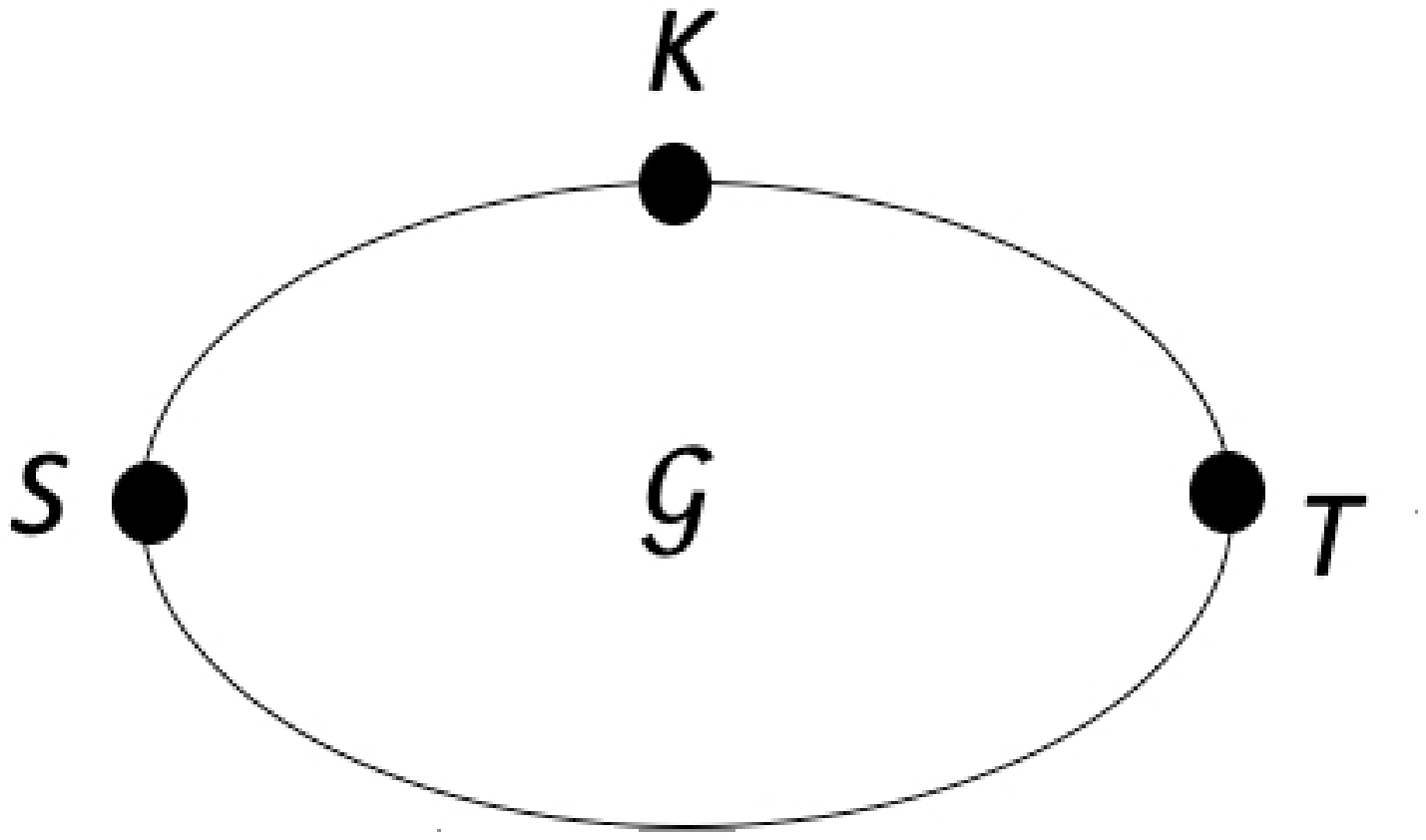}
\label{SDAG}}
\hspace{2mm}
\subfloat[]{\includegraphics[width=1.56in, height=1.1in]{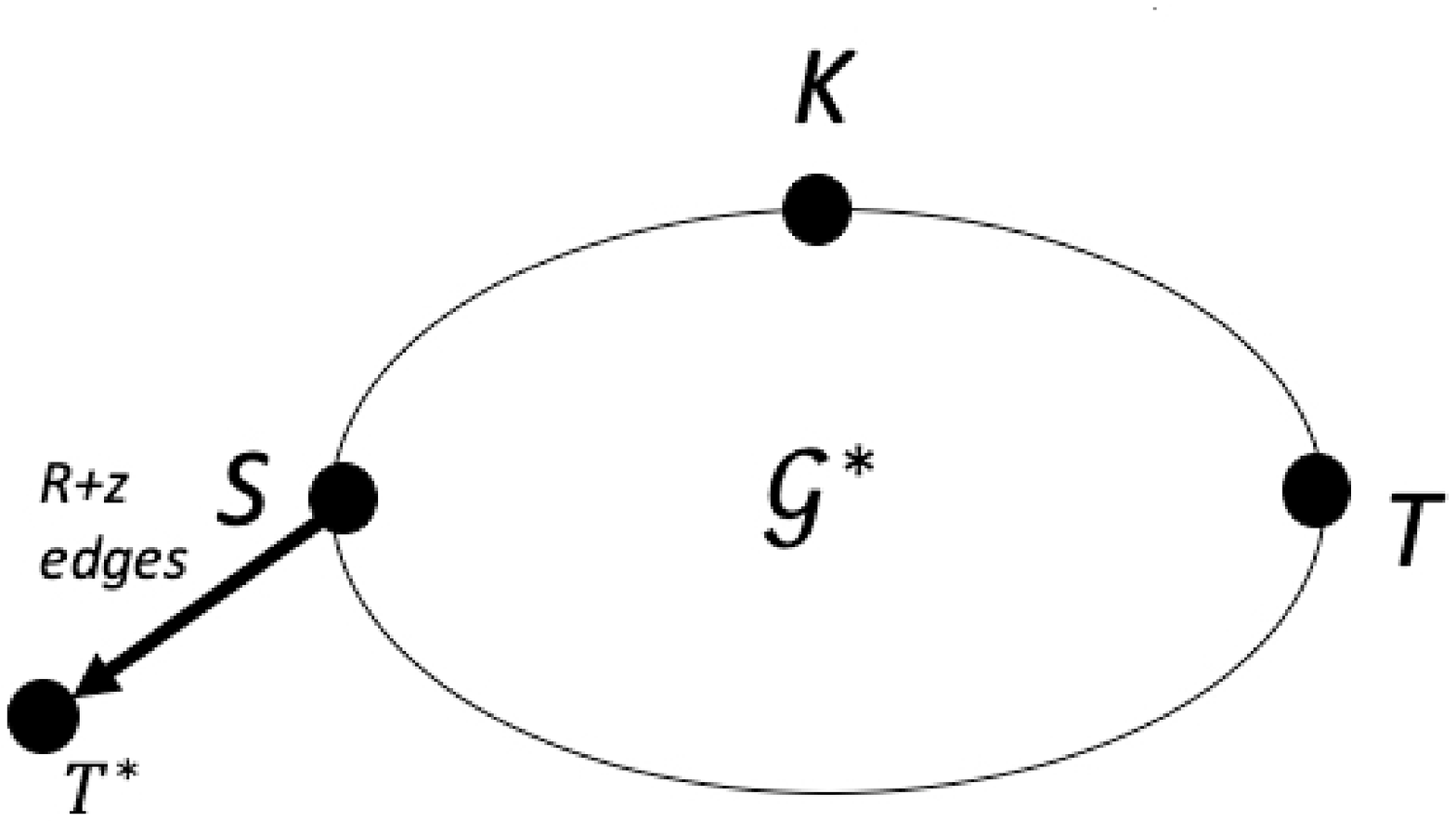}
\label{SDAGR}}
\vspace{0.0cm}
\caption{(a) Network model $\mathcal{G}$, and (b) the modified network $\mathcal{G}^*$ obtained from $\mathcal{G}$ by adding $T^*$ and setting the demands at $T$ and $T^*$ to $(M,N)$.}
\end{figure}

\subsection{Network Code}
\label{NC}
We define a scalar linear network code $\mathcal{N}$ for the network $\mathcal{G}$ to be an assignment of a linear encoding function $f_e$ to each edge $e \in \mathcal{E}$ and a linear decoding function $g_T$ to terminal $T$. For $e \in \mathcal{E}$, we denote the edge message on $e$ by $X_e$, and for any set $\mathcal{A} \subseteq \mathcal{E}$, we define $X_\mathcal{A} = \{X_e: e \in \mathcal{A}\}$. If $e \in \mathcal{E}$ and $e = (u,v)$ then the edge message $X_e$ is a linear combination of all the messages carried by the edges in $\In(u) = \{(w,u): (w,u) \in \mathcal{E}\}$, the incoming edges of $u$. The edge message at $e$ is obtained using local encoding at $u$. We define $X_e$ using the local encoding function $\bar{f}_e$ on $e = (u,v)$ as
\begin{align}
    \label{NC_EQ1a}
    X_e = \bar{f}_e(X_{\In(u)}) = \sum_{e' \in \In(u)}\bar{c}_{e',e}X_{e'}.
\end{align}
Here, $X_e$ denotes the message on edge $e$, for each edge $e' \in \In(u)$, $X_{e'}$ denotes the messages on edges $e'$ and $\bar{c}_{e',e}$ is the coefficient acting on each message $X_{e'}$. If edge $e$ is an outgoing edge of $S$ (or $K$), then $X_e$ is a function of the source messages (or keys) as well. Given, such a network code, an adversary that wiretaps any size-$z$ subset of edges $\mathcal{W} \subset \mathcal{E}$ would obtain the information $X_\mathcal{W}$ on the wiretapped edges. A network code is said to be $(R,z)$-feasible if
\begin{align}
    \label{NC_EQ1}
    g_T(X_{\In(T)}) &= M\\
    \label{NC_EQ2}
    \I(M;X_\mathcal{W}) &= 0,
\end{align}
where $T$ is the terminal node and $M$ is the $R$-dimensional message vector generated by the source $S$.

\subsection{Topological Order}
\label{NDCOMP}
To achieve secure communication over the network $\mathcal{G}$, the source $S$ must ``mix'' the message symbols in $M$ with the (received) random key symbols in $N$. This mixture of messages and keys is communicated to the terminal $T$, which must decode correctly to reconstruct message $M$. Let $\mathcal{V} = \{v_0,...,v_{n-1}\}$. Since $\mathcal{G}$ is directed and acylic, we assume, without loss of generality, that the nodes $v_i \in \mathcal{V}$ are indexed according to their topological order in $\mathcal{G}$. This implies that the node $v_i$ receives its incoming information only from nodes $v_0, \cdots, v_{i-1}$.
We also assume that the index of $K$ in this topological order is less than that of $S$ which in turn is less than that of the terminal $T$. More specifically, we assume $K=v_0$, $S=v_m$, and $T=v_{n-1}$ for $v_0, v_m, v_{n-1} \in \mathcal{V}$ and $0 < m < n-1$. There is no loss of generality in these assumptions as otherwise, either transmissions on outgoing edges of $S$ cannot be secure or the communication rate $R$ between $S$ and $T$ is zero. This implies that nodes $\{v_0,\dots,v_{m-1}\}$ only transmit, on their outgoing edges, functions of the information generated by $K$ while nodes $\{v_m,\dots,v_{n-1}\}$ may potentially transmit functions of the information generated at both $S$ and $K$.

\subsection{The Cut Sets}
\label{NPROP}


 
For any pair of nodes $u,v \in \mathcal{V}$, a cut is a set of edges in $\mathcal{E}$ which, when removed, disconnects all paths from $u$ to $v$. The cut with the minimum capacity that separates $u$ and $v$ is denoted as $\mincut_\mathcal{G}(u,v)$. Since each edge in $\mathcal{E}$ is assumed to be of unit capacity, $|\mincut_\mathcal{G}(u,v)|$ represents the total capacity of all the edges in $\mincut_\mathcal{G}(u,v)$. 
    The cuts as defined above may also separate sets of nodes in the network $\mathcal{G}$. For a subset of nodes $\mathcal{A}$, the set $\mincut_{\mathcal{G}}(\mathcal{A},v)$ is the minimum capacity cut that separates the set of nodes in $\mathcal{A} \subset \mathcal{V}$ from the node $v \in \mathcal{V}$.
    For the network $\mathcal{G}$, we use the following notation
    \begin{align}
        C_{K-S} &= |\mincut_\mathcal{G}(K,S)|\nonumber\\
        C_{K-T} &= |\mincut_\mathcal{G}(K,T)|\nonumber\\
        C_{S-T} &= |\mincut_\mathcal{G}(S,T)|\nonumber\\
        C_{KS-T} &= |\mincut_\mathcal{G}(\{K,S\},T)|\nonumber
    \end{align}


\section{Results}
\label{RES}

In this work we prove the following theorem.

\begin{theorem}
\label{TH1}
Given the directed acyclic network $\mathcal{G}$ and integers $R$ and $z$ such that $R > 0$, there exists an $(R,z)$-feasible network code $\mathcal{N}$ over $\mathcal{G}$ if and only if, 
\begin{align}
    \label{BND1}
    &z \leq \min(C_{K-S},C_{K-T})\\
    \label{BND2}
    &R \leq C_{S-T}\\
    \label{BND3}
    &R+z \leq C_{KS-T}
\end{align}
\end{theorem}

The proof of Theorem \ref{TH1} is divided into two parts, the achievability proof, shown in Section \ref{TH1_ACH}, and the converse proof shown in Section \ref{TH1_CONVa_body}.

\section{Proof of Theorem \ref{TH1}: Achievability}
\label{TH1_ACH}
\begin{proof}
For the network $\mathcal{G} = (\mathcal{E},\mathcal{V})$ with source node $S$ and key generating node $K$ holding $R$ message symbols $M$ and $z$ key symbols $N$ respectively, we set the values of integers $R$ and $z$ such that they satisfy the bounds (\ref{BND1}), (\ref{BND2}), and (\ref{BND3}). We implement a random linear network code $\mathcal{N}$ over $\mathcal{G}$ and over a sufficiently large field $\mathbb{F}_q$ such that, for any edge $e = (u,v) \in \mathcal{E}$, the local encoding coefficients $\{\bar{c}_{e',e}\}_{e' \in \In(u)}$ associated with edge $e$, as described in (\ref{NC_EQ1a}), are i.i.d. and uniform over $\mathbb{F}_q$.

The network code $\mathcal{N}$ is said to be decodable at rate $R$ over network $\mathcal{G}$, if it satisfies the condition of (\ref{NC_EQ1}). We consider the following lemma which we prove in Section \ref{DEC}.

\begin{lemma}
 \label{LEM_DEC}
 Given integers $R,z$ that satisfy (\ref{BND1})-(\ref{BND3}) of Theorem \ref{TH1}, the random linear network coding scheme $\mathcal{N}$ is decodable at rate $R$ with probability at least $1 - \dfrac{2(|\mathcal{E}|+R+z)^2}{q}$.
\end{lemma}

We now consider a wiretapping adversary that can eavesdrop on any subset of edges $\mathcal{W} \subset \mathcal{E}$ such that $|\mathcal{W}| = z$. We denote the information gleaned by the adversary as $X_{\mathcal{W}}$ which may be expressed as 
 \begin{align}
     \label{SEC_EQ9}
     X_{\mathcal{W}} = \irow{\mathbf{A}_\mathcal{W} & \mathbf{B}_\mathcal{W}}\icol{M \\ N}
 \end{align}
 Here, $\mathbf{A}_\mathcal{W}$ and $\mathbf{B}_\mathcal{W}$ are $z \times R$ and $z \times z$ matrices whose rows are global encoding vectors associated with each edge in $\mathcal{W}$, acting on $M$ and $N$, respectively. We consider the network coded information to be secure if and only if (\ref{NC_EQ2}) holds for any $\mathcal{W} \subset \mathcal{E}$ of size $z$, i.e. the adversary gains no information about the source message symbols $M$ even after wiretapping a $z$-sized subset of edges in the network. In \cite{cai2007security}, Cai and Yeung show that a linear network coding scheme is secure if and only if the following condition holds.
  \begin{align}
     \label{SEC_EQ10}
     \rk(\irow{\mathbf{A}_\mathcal{W} & \mathbf{B}_\mathcal{W}}) = \rk(\mathbf{B}_\mathcal{W})
 \end{align}
 Here, $\rk(.)$ denotes the {\em rank} of a matrix.
 
 The following lemma is proven in Section \ref{SEC} by analyzing the matrices $\mathbf{A}_\mathcal{W}$ and $\mathbf{B}_\mathcal{W}$.
 
 \begin{lemma}
 \label{LEM_SEC}
  Given integers $R,z$ that satisfy (\ref{BND1})-(\ref{BND3}) of Theorem \ref{TH1}, the random linear network coding scheme $\mathcal{N}$ over $\mathcal{G}$ is $z$-secure with probability at least $1 - \dfrac{\binom{|\mathcal{E}|}{z}2z}{q}$ for all wiretap sets $\mathcal{W} \subset \mathcal{E}$ of size $z$.
\end{lemma}


A network code is said to be $(R,z)$-feasible if it is both $R$-feasible and $z$-secure. It now follows that, given integers $R$ and $z$ that satisfy (\ref{BND1}), (\ref{BND2}), and (\ref{BND3}), the suggested network code is $(R,z)$-feasible with probability at least 
$$\Big(1 - \dfrac{2(|\mathcal{E}|+R+z)^2 + \binom{|\mathcal{E}|}{z}2z}{q} \Big),$$
which, for sufficiently large $q$, implies our achievability with high probability.

\end{proof}

\section{Proof of Theorem \ref{TH1}: Converse}
\label{TH1_CONVa_body}
\begin{proof}
We prove the converse for any (not necessarily linear) $(R,z)$-feasible network code $\mathcal{N}$ over the network $\mathcal{G}$. We start with an $(R,z)$-feasible coding scheme and show that $R$ and $z$ satisfy the bounds of (\ref{BND1}), (\ref{BND2}) and (\ref{BND3}). 
Here, we give a partial proof in which we only address bound (\ref{BND1}). Proofs of a similar nature apply to the other bounds as well. Details of the converse proof, in the more general context of cyclic networks, appear in Appendix \ref{TH1_CONV}.

We denote by $\mathbb{C}_{K-S}$ the minimum cut separating $K$ and $S$, and by $C_{K-S}$ the total capacity of the edges in $\mathbb{C}_{K-S}$. The random variable $X_{K-S}$, over the support set $\mathcal{X}_{K-S}$, represents the information on all  edges of $\mathbb{C}_{K-S}$. We denote by $\mathcal{W}$ any subset of $z$ edges in $\mathcal{E}$ that is wiretapped by an eavesdropping adversary. Then $X_\mathcal{W}$ denotes the encoded information on all the edges in $\mathcal{W}$. We denote the set of edges that are incoming to $S$ as $\In(S)$, and the encoded information on all of the edges in $\In(S)$ as $X_{\In(s)}$ with support set $\mathcal{X}_{\In(S)}$. Similarly, for $\Out(S)$.

For the bound $z \leq \min(C_{K-S},C_{K-T})$ we consider two cases.
First, assume by contradiction that $z > C_{K-S}$. Specifically set $z = C_{K-S} + 1$. This implies that the eavesdropping adversary may choose to wiretap all the edges in $\mathbb{C}_{K-S}$ and an edge $e \in \Out(S)$ to obtain the wiretap set $\mathcal{W} = \mathbb{C}_{K-S} \cup \{e\}$ of size $z$. Then the wiretapped information is $X_\mathcal{W} = (X_{K-S},X_e)$, where $X_e$ is the information on the chosen edge $e$. Note that $X_e = \bar{f}_{e}(X_S)$, 
where, $X_{S} := (M, X_{\In(S)})$ is the information present at the source $S$. 

For $z$-security, we require that the mutual information $\I(M;X_\mathcal{W}) = 0$. Therefore,
\begin{align}
    \I(M;X_\mathcal{W}) &= \I(M;X_{K-S}) + \I(M;X_e|X_{K-S}) = 0,\nonumber
\end{align}
implying that,
$\I(M;X_{K-S}) = 0$ and 
$\I(M;X_e|X_{K-S}) = 0$.
Thus, we conclude that  
$\Hi(X_e|X_{K-S}) = \Hi(X_e|X_{K-S},M)$.

Suppose that cut $\mathbb{C}_{K-S}$ partitions $\mathcal{G}$ into disjoint sub-networks $\mathcal{A}$ and $\bar{\mathcal{A}}$, where $\mathcal{A}$ includes the key generating node $K$. Note that any information communicated through edges in $\bar{\mathcal{A}}$ must be a function of $X_{K-S}$. In addition, $\In(S) \subset \mathbb{C}_{K-S} \cup \mathcal{E}_{\bar{\mathcal{A}}}$, implying that all information reaching $S$ is a function of $X_{K-S}$. We conclude, for any edge $e \in \Out(S)$, that
\begin{align}
    \label{TH_CONV_EQ6a_1a}
    X_e &= h_e(M,X_{K-S}),
\end{align}
where, $h_e$ is some deterministic function. Equation (\ref{TH_CONV_EQ6a_1a}) implies that $\Hi(X_e|X_{K-S},M) = 0$ which in turn implies $\Hi(X_e|X_{K-S}) = 0$. This means that to be $z$-secure the information $X_{K-S}$ must completely determine $X_e$ for all $e \in \Out(S)$. Therefore, the information $X_{\Out(S)} := \{X_e\}_{e \in \Out(S)}$ is also a deterministic function of $X_{K-S}$. As $\I(M;X_{K-S}) = 0$ shows that $X_{K-S}$ is independent of $M$, it follows that $X_{\Out(S)}$ is also independent of $M$ and thus $\I(M;X_{\Out(S)}) = 0$. This, in turn, implies that the rate realizable by the network code $\mathcal{N}$ is $R=0$ which is a contradiction.

A similar proof holds for $z \leq C_{K-T}$, in which we study the set $\mathcal{W} = \mathbb{C}_{K-T} \cup \{e\}$ for any edge $e \in \In(T)$.
\end{proof}

\section{Proof of Lemmas}
\label{LEMPROOFS}
\subsection{Proof of Lemma \ref{LEM_DEC}}
\label{DEC}
\begin{lemmaproof}
We begin by considering the modified network $\mathcal{G}^* = (\mathcal{V}^*,\mathcal{E}^*)$, obtained from $\mathcal{G}$ as shown in Figure \ref{SDAGR}. Specifically, $\mathcal{G}^*$ is obtained from $\mathcal{G}$ by adding a new node $T^*$ and $R+z$ parallel edges from $S$ to $T^*$. As in $\mathcal{G}$, the network $\mathcal{G}^*$ has nodes $S$ and $K$ holding $R$ symbols of $M$ and $z$ symbols of $N$, respectively. Here, the outgoing edges of $S$ include those in the original network $\mathcal{G}$, denoted as $\Out(S)$, and the additional $R+z$ edges. Both terminals $T$ and $T^*$ want to decode all $R$ symbols of $M$ and $z$ symbols of $N$. A network code, over $\mathcal{G}^*$, that satisfies the demands of terminals $T$ and $T^*$ is a multi-source multicast network code which is $\mathbf{R}$-feasible, where $\mathbf{R} = (R,z)$.

We use a random linear multi-source multicast network code $\mathcal{N}^*$ over network $\mathcal{G}^*$ and the finite field $\mathbb{F}_q$. In what follows, we set some notation.
\begin{itemize}
    \item [1.] Let $O_K \triangleq |\Out(K)|$, $I_S \triangleq |\In(S)|$ and $O_S \triangleq |\Out(S)|$.
    \item [2.] The node $K$ transmits $z$ linear combinations of $N$ through $\Out(K)$. We express the information on these edges as $X_{\Out(K)} = \mathbf{B}_KN$. Here, the rows of $\mathbf{B}_K$, which is an $O_K \times z$ matrix, are the local encoding vectors associated with each edge in $\Out(K)$. The entries of $\mathbf{B}_K$ are i.i.d. and uniform over the field $\mathbb{F}_q$.
    \item [3.] The message source $S$ receives $I_S$ linear combinations of $N$ through the edges in $\In(S)$. We express the information on these edges as $X_{\In(S)} = \mathbf{V}_{In(S)}\mathbf{B}_KN$. $\mathbf{V}_{\In(S)}$ is an $I_S \times O_K$ matrix, and the rows of $\mathbf{V}_{In(S)}\mathbf{B}_K$ are the global encoding vectors, associated with each edge in $\In(S)$, acting on $N$.
    \item [4.] $S$ ``mixes" the received $I_S$ symbols of $X_{\In(S)}$ with the $R$ symbols of $M$ and transmits the resulting combinations through $\Out(S)$ and to $T^*$. We express the information on $\Out(S)$ as 
    \begin{align}
         X_{\Out(S)} &= \irow{\mathbf{A_S} & \mathbf{B}_S}\icol{M \\ \mathbf{V}_{In(S)}\mathbf{B}_KN}\nonumber\\
        &= \irow{\mathbf{A_S} & \mathbf{B}_S\mathbf{V}_{In(S)}\mathbf{B}_K}\icol{M \\ N}.\nonumber
    \end{align}
    Here, the rows of the matrix $\irow{\mathbf{A_S} & \mathbf{B}_S}$ are the local encoding vectors associated with the edges in $\Out(S)$. $\mathbf{A}_S$ and $\mathbf{B}_S$ are $O_S \times R$ and $O_S \times I_S$ matrices respectively. The entries of $\mathbf{A}_S$ and $\mathbf{B}_S$ are i.i.d. and uniform over $\mathbb{F}_q$.
\end{itemize}

We now consider the following claims. Claim \ref{CL_RED} is proven in Appendix \ref{CLAIMPROOF_CL_RED}.
\begin{claim}
\label{CL_MC}
The multi-source multicast random linear network code $\mathcal{N}^*$, as described above, is $\mathbf{R}$-feasible over the network $\mathcal{G}^*$ with probability at least $ 1 - \dfrac{2(|\mathcal{E}|+R+z)^2}{q}$. 
\end{claim}

\begin{proof}[Proof of Claim \ref{CL_MC}]
\label{CLAIMPROOF_CL_MC}
Given integers $R$ and $z$, we start by observing the min-cut capacities in $\mathcal{G}^*$ between the subsets of the node set $\{S,K\}$ and each terminal $T$ and $T^*$ as follows.
\begin{align}
    \label{CLAIMPROOF_CL_MC_EQ1}
    &|\mincut_{\mathcal{G}^*}(K,T)| = C_{K-T} \geq z\\
    \label{CLAIMPROOF_CL_MC_EQ2}
    &|\mincut_{\mathcal{G}^*}(K,T^*)| = \min(R+z,C_{K-S}) \geq z \\
    \label{CLAIMPROOF_CL_MC_EQ3}
    &|\mincut_{\mathcal{G}^*}(S,T^*)| = R+z \geq R\\
    \label{CLAIMPROOF_CL_MC_EQ4}
    &|\mincut_{\mathcal{G}^*}(S,T)| = C_{S-T} \geq R\\
    \label{CLAIMPROOF_CL_MC_EQ5}
    &|\mincut_{\mathcal{G}^*}(\{K,S\},T^*)| = R+z \\
    \label{CLAIMPROOF_CL_MC_EQ6}
    &|\mincut_{\mathcal{G}^*}(\{K,S\},T)| = C_{KS-T} \geq R+z
\end{align}

From (\ref{CLAIMPROOF_CL_MC_EQ1})-(\ref{CLAIMPROOF_CL_MC_EQ6}), we see that for all source-terminal pairs in $\mathcal{G}^*$, the corresponding Min-Cut Max-Flow bounds are satisfied.

Let $L$ be the total number of encoding coefficients employed over all the edges in $\mathcal{E}^*$. We can bound $L$ by $\sum_{e \in \mathcal{E}^*}|\mathcal{E}^*| \leq |\mathcal{E}^*|^2 = (|\mathcal{E}|+R+z)^2$. Using Theorem 8 of \cite{koetter2003algebraic} and Theorem 5.4 of \cite{8187170} (derived from \cite{ho2006random}), we have that the network code $\mathcal{N}^*$ is $\mathbf{R}$-feasible over the network $\mathcal{G}^*$ with probability at least
\begin{align}
    \Big(1 - \dfrac{2}{q}\Big)^L &> 1 - \dfrac{2L}{q}
    > 1 - \dfrac{2(|\mathcal{E}|+R+z)^2}{q}\nonumber
\end{align}

This proves the claim.

\end{proof}

\begin{claim}
\label{CL_RED}
The $\mathbf{R}$-feasible network code $\mathcal{N}^*$ over $\mathcal{G}^*$, when restricted to $\mathcal{G}$, implies that $\mathcal{N}$ is $R$-decodable over $\mathcal{G}$.
\end{claim}


From Claim \ref{CL_MC} and Claim \ref{CL_RED}, we have that the network code $\mathcal{N}$ is $R$-decodable over $\mathcal{G}$ with probability at least
\begin{align}
    1 - \dfrac{2(|\mathcal{E}|+R+z)^2}{q}\nonumber
\end{align}

This proves the lemma.
\end{lemmaproof}

\subsection{Proof of Lemma \ref{LEM_SEC}}
\label{SEC}
\begin{lemmaproof}
We use the notation introduced in the proof of Lemma \ref{LEM_DEC}. For any edge $e \in \mathcal{E}$, we express the information on $e$ as,
\begin{align}
    \label{SEC_EQ3a}
    X_{e} &= u_{e}\icol{X_{\Out(K)} \\ X_{\Out(S)}} = u_{e}\begin{bmatrix}
    \mathbf{0} & \mathbf{B}_K\\
    \mathbf{A}_S & \mathbf{B}_S\mathbf{V}_{\In(S)}\mathbf{B}_K
    \end{bmatrix}\icol{M \\ N}
\end{align}
Here, $u_{e}$ is an edge-$e$ encoding vector of dimension $O_K+O_S$, acting on $X_{\Out(K)}$ and $X_{\Out(S)}$. We partition $u_e = \irow{u_K & u_S}$ such that the $O_K$-dimensional vector $u_K$ acts on the information from $\Out(K)$ and the $O_S$-dimensional vector $u_S$ acts on the information from $\Out(S)$. Thus, we rewrite \eqref{SEC_EQ3a} as follows.
\begin{align}
    \label{SEC_EQ3b}
    X_{e} &= \irow{u_K & u_S}\begin{bmatrix}
    \mathbf{0} & \mathbf{B}_K\\
    \mathbf{A}_S & \mathbf{B}_S\mathbf{V}_{\In(S)}\mathbf{B}_K
    \end{bmatrix}\icol{M \\ N}
\end{align}

We now consider an adversary that wiretaps any subset $\mathcal{W} \subset \mathcal{E}$ of edges such that $|\mathcal{W}| = z$. Then, using \eqref{SEC_EQ3b}, we obtain the information observed by the adversary as follows.
\begin{align}
    \label{SEC_EQ4a}
    X_{\mathcal{W}} &= \irow{\mathbf{U}_K & \mathbf{U}_S}\icol{X_{\Out(K)} \\ X_{\Out(S)}}
\end{align}
Here, $\irow{\mathbf{U}_K & \mathbf{U}_S}$ is a $z \times (O_K+O_S)$ matrix where $\mathbf{U}_K$ is a $z \times O_K$ matrix and $\mathbf{U}_S$ is a $z \times O_S$ matrix. We assume that $\irow{\mathbf{U}_K & \mathbf{U}_S}$ has full row-rank of $z$, as otherwise, the adversary could simply drop an edge in $\mathcal{W}$ and not lose any information. Using \eqref{SEC_EQ3b}, we rewrite \eqref{SEC_EQ4a} as follows.
\begin{align}
    \label{SEC_EQ4}
    X_\mathcal{W} &= \begin{bmatrix}
    \mathbf{U}_S\mathbf{A}_S & \mathbf{U}_K\mathbf{B}_K + \mathbf{U}_S\mathbf{B}_S\mathbf{V}_{\In(S)}\mathbf{B}_K
    \end{bmatrix}\icol{M \\ N}
\end{align}
From \eqref{SEC_EQ9} and \eqref{SEC_EQ4}, we have that
\begin{align}
    \mathbf{A}_\mathcal{W} &= 
    \mathbf{U}_S\mathbf{A}_S \quad \text{and} \quad
    \mathbf{B}_\mathcal{W} = \begin{bmatrix}
    \mathbf{U}_K + \mathbf{U}_S\mathbf{B}_S\mathbf{V}_{\In(S)}
    \end{bmatrix}\mathbf{B}_K\nonumber
\end{align}
Let,
\begin{align}
    \label{SEC_EQ6}
    \mathbf{\Phi} &\triangleq \irow{
    \mathbf{U}_K + \mathbf{U}_S\mathbf{B}_S\mathbf{V}_{\In(S)}
    }.
\end{align}
From our decodability proof, we know that $\rk(\mathbf{V}_{\In(S)}) = z$, as otherwise, $T^*$ could not have decoded the keys $N$. For the security condition of \eqref{SEC_EQ10} to hold, we show that $\rk(\mathbf{B}_\mathcal{W}) = \rk(\mathbf{\Phi}\mathbf{B}_K) = z$. Therefore, we compute the following.
\begin{align}
    \label{SEC_EQ7}
    &\Pr_{\mathbf{B}_K,\mathbf{B}_S}\{\rk(\mathbf{B}_\mathcal{W}) = z\} = \nonumber\\
    &\quad \quad \Pr_{\mathbf{B}_S}\{\rk(\mathbf{\Phi}) = z\}
    \Pr_{\mathbf{B}_K}\{\rk(\mathbf{\Phi}\mathbf{B}_K) = z|\rk(\mathbf{\Phi}) = z\}.
\end{align}

We now consider the following claims proven in Appendix \ref{CLAIMPROOF_SEC_CL_1} and Appendix  \ref{CLAIMPROOF_SEC_CL_2}, respectively.

\begin{claim}
\label{SEC_CL_1}
$\Pr_{\mathbf{B}_S}\{\rk(\mathbf{\Phi}) = z\} > 1 - \dfrac{z}{q}$
\end{claim}

\begin{claim}
\label{SEC_CL_2}
Given an $n \times m$ matrix $\mathbf{A}$ and an $m \times n$ matrix $\mathbf{B}$ such that $\rk(\mathbf{A}) = n$ and the entries of $\mathbf{B}$ are i.i.d. and uniform over the field $\mathbb{F}_q$, then $\rk(\mathbf{A}\mathbf{B}) = n$ with probability at least $1 - \dfrac{n}{q}$, over $\mathbf{B}$.
\end{claim}

Let us consider the following event.
\begin{itemize}
    \item $\mathbb{E}_\mathcal{W}$: The condition of (\ref{SEC_EQ10}) holds for a given wiretap set $\mathcal{W}$ of size $z$.
\end{itemize}
Using Claim \ref{SEC_CL_1} and Claim \ref{SEC_CL_2} we conclude from \eqref{SEC_EQ7} that
\begin{align}
    \label{SEC_EQ_9}
    \Pr_{\mathbf{B}_K,\mathbf{B}_S}\{\mathbb{E}_\mathcal{W}\} &> \Big(1 -  \dfrac{z}{q}\Big)^2 > 1 -  \dfrac{2z}{q}
\end{align}
Denoting the complementary event of $\mathbb{E}_\mathcal{W}$ by $\Bar{\mathbb{E}}_\mathcal{W}$ and using the union bound over event $\Bar{\mathbb{E}}_\mathcal{W}$ for any $\mathcal{W} \subset \mathcal{E}$ of size $z$, we have the following.
\begin{align}
    \Pr\{\bigcup_{\mathcal{W} \subset \mathcal{E}}\Bar{\mathbb{E}}_\mathcal{W}\} \quad & \leq \sum_{\mathcal{W} \subset \mathcal{E}}\dfrac{2z}{q} = \dfrac{\binom{|\mathcal{E}|}{z}2z}{q}.\nonumber
\end{align}
Namely, the probability over the i.i.d. entries of $\mathbf{B}_S$ and $\mathbf{B}_K$, of the network code being secure against an adversary with a wiretap set $\mathcal{W}$ of size $z$ is at least $1 - \dfrac{\binom{|\mathcal{E}|}{z}2z}{q}$. This proves the lemma.
\end{lemmaproof}

\section{Conclusion}
\label{CONC}
In this paper, we characterize the capacity-security region for single unicast network codes over a directed acyclic network in which only one node, which is not necessarily the source node, can generate random keys. We present a random linear achievability proof and a matching coverse proof. Our converse can be extended to cyclic networks as well. (Details appear in Appendix \ref{TH1_CONV}.) Our work establishes an intermediate step between the well understood problem of characterizing the capacity-security region in which only the source node generates random keys and the problem of characterizing the capacity-security region when every node can generate random keys. 

Several problems are left open. An extension of our result to the context of multicast network coding is within reach and the subject of future research. It would also be interesting to extend our achievability to single unicast network coding over networks with cycles. Additional possible extensions include the study of single unicast networks in which more than one node can independently generate random keys. 

\section*{Acknowledgements}
Work supported in part by NSF grants CCF-1526771 and CCF-1817241.

\bibliographystyle{IEEEtran}
\bibliography{ref}

\begin{appendices}
\section{Proof of Theorem \ref{TH1}: Converse (For cyclic networks)}
\label{TH1_CONV}
\begin{proof}
We prove the converse for the more general setting of directed networks $\mathcal{G}_{\rm cyc} = (\mathcal{V}_{\rm cyc},\mathcal{E}_{\rm cyc})$ that may contain cycles.
As before, $\mathcal{G}_{\rm cyc}$ has the message generating source node $S$, the random key generating node $K$ and the terminal $T$, as shown in Figure \ref{CS1}. We show that for such a network and for any network coding scheme, the bounds given in (\ref{BND1}), (\ref{BND2}) and (\ref{BND3}) are upper bounds for the capacity-security region.

As we address networks with cycles, we consider the notion of {\em time} in our definition of a network code.
Namely, we consider an $n$-time step system. In such a system, we assume that communication starts at time step $i=1$. The source $S$ holds a message $M$ uniformly distributed in $[q^{nR}]$ and node $K$ holds random keys $N$ uniformly distributed in $[q^{R_K}]$, where $R_K$ is not restricted in any way. For any edge $e \in \mathcal{E}$ such that $e = (u,v)$, where $u,v \in \mathcal{V}_{\rm cyc}$, we define the information on $e$ at the $i$-th time step, for all $i \in [n] = \{1,\cdots,n-1\}$, as 
\begin{align}
\label{CYC_2}
X_e^{(i)} &= \bar{f}_e^{(i)}(\{X_{e'}^{(j)}\}_{e' \in \In(u), j \in [i-1]})
\end{align}

Here, $\bar{f}_e^{(i)}$ is the time-variant local encoding function at edge $e$ at the $i$-th time step, $\In(u)$ denotes the set of incoming edges in $\mathcal{E}$ to node $u$ and $[i-1] = \{1,\cdots,i-1\}$. In our work, we consider $X_e^{(i)}$ to be a random variable with the support set $\mathcal{X}_e$, for all $i \in [n]$. For a given cut $\mathbb{C}$, we denote by $X_\mathbb{C}^{(i)}$ the composite of the variables corresponding to edges $e \in \mathbb{C}$ at time step $i$, i.e. $X_\mathbb{C}^{(i)} = (\{X_e^{(i)}\}_{e \in \mathbb{C}})$. The support set of $X_\mathbb{C}^{(i)}$ for all $i \in [n]$ is denoted by $\mathcal{X}_\mathbb{C}$. We use the notation $X_{e}^{[n]}$ to denote the information on edge $e$ for the $n$-time step system, i.e. $X_{e}^{[n]} := \{X_e^{(j)}\}_{j \in [n]}$.

For the network model $\mathcal{G}_{\rm cyc}$, the following definitions are useful for the discussions that follow:
\begin{itemize}
\item For any cut  $\mathbb{C}_{u-v}$, separating any two nodes $u,v \in \mathcal{V}$, we define two sub-networks $\mathcal{A} = (\mathcal{V}_\mathcal{A},\mathcal{E}_\mathcal{A})$ and $\bar{\mathcal{A}} = (\mathcal{V}_{\bar{\mathcal{A}}},\mathcal{E}_{\bar{\mathcal{A}}})$, with $u \in \mathcal{V}_\mathcal{A}$ and $v \in \mathcal{V}_{\bar{\mathcal{A}}}$, as shown in Figure \ref{CS1}. Here, $\mathcal{V} = \mathcal{V}_\mathcal{A} \cup \mathcal{V}_{\bar{\mathcal{A}}}$ and $\mathcal{E} = \mathcal{E}_\mathcal{A} \cup \mathcal{E}_{\bar{\mathcal{A}}} \cup \mathbb{C}_{u-v}$.
\item We denote by $\mathbb{C}_{K-S}$, the minimum cut separating $K$ and $S$ and by $C_{K-S}$, the total capacity of the edges in $\mathbb{C}_{K-S}$. The random variable $X_{K-S}^{(i)}$, over the support set $\mathcal{X}_{K-S}$, represents the information, at the $i$-th time step on all the edges of $\mathbb{C}_{K-S}$. $X_{K-S}^{[n]} = \{X_{K-S}^{[j]}\}_{j \in [n]}$ represents the information on all the edges of $\mathbb{C}_{K-S}$ for the $n$-time step system. We use similar notations for the time variant random variables which represent the information on the edges in $\mathbb{C}_{K-T}$, $\mathbb{C}_{S-T}$, and $\mathbb{C}_{KS-T}$.

\item We denote by $\mathcal{W}$ any subset of $z$ edges in $\mathcal{E}$ that is wiretapped by an eavesdropping adversary. $X_\mathcal{W}^{(i)}$ denotes the encoded information on all the edges in $\mathcal{W}$ at the $i$-th time step. We assume that the wiretap set $\mathcal{W}$ is time invariant, i.e. it does not change with the time step $i$. Thus, the information obtained by the adversary for the $n$-time step system is $X_{\mathcal{W}}^{[n]}$. 

\item We denote by $\In(S)$, the set of edges that are incoming to $S$. We denote the encoded information at the $i$-th time step on all the edges in $\In(S)$ as $X_{\In(s)}^{(i)}$ with support set $\mathcal{X}_{\In(S)}$. Similarly, for $\Out(S)$.

\item For any sub-graph $\mathcal{A} = (\mathcal{V}_\mathcal{A},\mathcal{E}_\mathcal{A}) \subset \mathcal{G}_{\rm cyc}$, let 

\begin{align}
    \label{SRC_DEF}
    \mathcal{S}_\mathcal{A} := \{S,K\} \cap \mathcal{V}_\mathcal{A}.
\end{align}
\end{itemize}

Given the definitions above, we start with an $(R,z)$-feasible coding scheme and show that $R$ and $z$ satisfy the bounds of (\ref{BND1}), (\ref{BND2}) and (\ref{BND3}). Here, a scheme is $(R,z)$-feasible with blocklength $n$ if $T$ can decode $M \in [q^{nR}]$ and any wiretapped subset of edge $\mathcal{W}$ hold no information on $M$. We shall now consider each of the bounds separately in the following subsections.

\begin{figure}[!t]
\centering
\includegraphics[width=2.1in, height=1.7in]{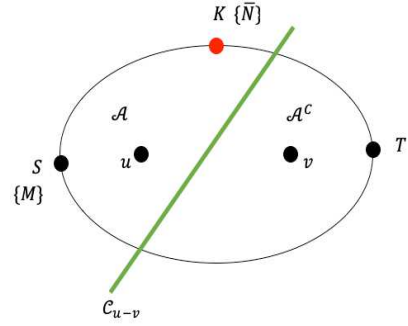}
\vspace{-0.0cm}
\caption{The partitioning of a network $\mathcal{G}$ due to the cut set $\mathbb{C}_{u-v}$.}
\label{CS1}
\end{figure}

\subsection{Bound on $z$: $z \leq \min(C_{K-S},C_{K-T})$}

\subsubsection{$\mathbf{z \leq C_{K-S}}$}

Suppose, by contradiction, that $z > C_{K-S}$. Particularly, assume $z = C_{K-S}+1$ This implies that the eavesdropping adversary may choose to wiretap all the edges in $\mathbb{C}_{K-S}$ and an edge $e \in \Out(S)$ to obtain the wiretap set $\mathcal{W} = \mathbb{C}_{K-S} \cup \{e\}$ of size $z$. Then the wiretapped information is $X_\mathcal{W}^{[n]} = (X_{K-S}^{[n]},X_e^{[n]})$, where $X_e^{[n]}$ is the information on the chosen edge $e$ for the $n$-time step system. From (\ref{CYC_2}), we see that
\begin{align}
\label{TH_CONV_EQ1}
X_e^{(i)} &= \bar{f}_e^{(i)}(\{X_S^{(j)}\}_{j \in [i-1]})\nonumber\\
&= \bar{f}_e^{(i)}(\{X_S^{[i-1]}\}).
\end{align}

Where, $X_{S}^{[i-1]} := (M, X_{\In(S)}^{[i-1]})$ is the information present at the source $S$ for all time steps up to the $(i-1)$-th time-step. For $z$-security, we require that the mutual information $\I(M;X_\mathcal{W}^{(i)}) = 0$. Therefore,

\begin{align}
\label{TH_CONV_EQ2}
\I(M;X_\mathcal{W}^{[n]}) &= \I(M;X_{K-S}^{[n]}) + \I(M;X_e^{[n]}|X_{K-S}^{[n]}) = 0.
\end{align}

Implying that,
\begin{align}
\label{TH_CONV_EQ3} 
\I(M;X_{K-S}^{[n]}) &= 0,\\
\label{TH_CONV_EQ4}
\I(M;X_e^{[n]}|X_{K-S}^{[n]}) &= 0.
\end{align}

From (\ref{TH_CONV_EQ4}), we obtain the following.

\begin{align}
    \label{TH_CONV_EQ5}
    \Hi(X_e^{[n]}|X_{K-S}^{[n]}) = \Hi(X_e^{[n]}|X_{K-S}^{[n]},M).
\end{align}

Suppose the cut $\mathbb{C}_{K-S}$ partitions $\mathcal{G}_{\rm cyc}$ into disjoint sub-networks $\mathcal{A}$ and $\bar{\mathcal{A}}$ . Then, as per the definition in (\ref{SRC_DEF}), $\mathcal{S}_{\bar{\mathcal{A}}} = \{S\}$. We denote by $X_{\bar{\mathcal{A}}}$, the source message $M$ and/or key $N$ held by the nodes in $\bar{\mathcal{A}}$. We see that $\In(S) \subset \mathbb{C}_{K-S} \cup \mathcal{E}_{\bar{\mathcal{A}}}$, which implies that any edge $e' \in \In(S)$ either belongs to the set $\In(S) \cap \mathbb{C}_{K-S}$ or the set $\In(S) \cap \mathcal{E}_{\bar{\mathcal{A}}}$.

For any edge $e' \in \In(S) \cap \mathbb{C}_{K-S}$, we observe that

\begin{align}
    \label{TH_CONV_EQ6a}
    X_{e'}^{[n]} = h_{e'}(X_{K-S}^{[n]}),
\end{align}

For $e' \in \In(S) \cap \mathcal{E}_{\bar{\mathcal{A}}}$, we consider the following lemma which we prove in Appendix \ref{CS_LEM}.





\begin{lemma}
\label{LEM_CONV}
For any cut $\mathbb{C}$ that partitions graph $\mathcal{G}_{\rm cyc}$ into disjoint sub-networks $\mathcal{A}$ and $\bar{\mathcal{A}}$, there exists, for any edge $e \in \mathcal{E}_{\bar{\mathcal{A}}}$ and any time step $i \in [n]$, a deterministic mapping $g_e^{(i)}$ such that $g_e^{(i)}(\{X_\mathbb{C}^{(j)}\}_{j \in [i-1]},X_{\bar{\mathcal{A}}}) = X_e^{(i)}$. 
\end{lemma}

Therefore, using Lemma \ref{LEM_CONV} for edge $e' \in \In(S) \cap \mathcal{E}_{\bar{\mathcal{A}}}$:
\begin{align}
    \label{TH_CONV_EQ7_1}
    X_{e'}^{[j]} &= g_{e'}^{(j)}(M, X_{K-S}^{[j-1]})\nonumber\\
    &= \bar{h}_{e'}(M, X_{K-S}^{[j-1]}),
\end{align}

where, $\bar{h}_{e'}$ is a deterministic function.

Then, using (\ref{TH_CONV_EQ7_1}) we obtain the information on $\In(S)$ as follows.

\begin{align}
     \label{TH_CONV_EQ7b}
     X_{\In(S)}^{[j-1]} &= \{X_{e'}^{[j-1]}\}_{e' \in \In(S) \cap \mathbb{C}_{K-S}} \cup \{X_{e'}^{[j-1]}\}_{e' \in \In(S) \cap \mathcal{E}_{\bar{\mathcal{A}}}}\nonumber\\
     &=\{h_{e'}(X_{K-S}^{[j-1]})\}_{e' \in \In(S) \cap \mathbb{C}_{K-S}} \nonumber\\
     & \text{ } \cup \{\bar{h}_{e'}(M, X_{K-S}^{[j-2]})\}_{e' \in \In(S) \cap \mathcal{E}_{\bar{\mathcal{A}}}}\nonumber\\
     &= \bar{h}_{\In(S)}(M, X_{K-S}^{[j-1]}).
\end{align}

Thus, for the chosen edge $e \in \Out(S)$, using (\ref{TH_CONV_EQ7b}) and (\ref{TH_CONV_EQ1}) we have
\begin{align}
    \label{TH_CONV_EQ7}
    X_e^{[n]} &= \{X_{e}^{(j)}\}_{j \in [n]}\nonumber\\
    &= \{\bar{f}_e^{(i)}(M,X_{\In(S)}^{[j-1]})\}_{j \in [n]}\nonumber\\
    &= \{\bar{f}_e^{(i)}(M,\bar{h}_{\In(S)}(M, X_{K-S}^{[j-1]}))\}_{j \in [n]}\nonumber\\
    &= f(M,X_{K-S}^{[n-1]}).
\end{align}

Thus, (\ref{TH_CONV_EQ7}) shows that $X_e^{[n]}$ is a deterministic function of $M$ and $X_{K-S}^{[n-1]}$. As $\Hi(X_e^{[n]}|X_{K-S}^{[n]},M) \leq \Hi(X_e^{[n]}|X_{K-S}^{[n-1]},M)$, this implies that 
\begin{align}
    \label{TH_CONV_EQ9}
    \Hi(X_e^{[n]}|X_{K-S}^{[n]},M) = 0.
\end{align}

Therefore by (\ref{TH_CONV_EQ5}) and (\ref{TH_CONV_EQ9}), we have,
\begin{align}
    \label{TH_CONV_EQ10}
    \Hi(X_e^{[n]}|X_{K-S}^{[n]}) = 0.
\end{align}

Thus, to be $z$-secure, (\ref{TH_CONV_EQ10}) shows that for all $e \in \Out(S)$, the random variable $X_e^{[n]}$ must be completely determined by $X_{K-S}^{[n]}$. Therefore, the information $X_{\Out(S)}^{[n]} := \{X_e^{[n]}\}_{e \in \Out(S)}$ is also a deterministic function of $X_{K-S}^{[n]}$. As (\ref{TH_CONV_EQ3}) shows that $X_{K-S}^{[n]}$ is independent of the message symbols $M$, it follows that $X_{\Out(S)}^{[n]}$ is also independent of $M$ and thus
\begin{align}
\label{TH_CONV_EQ11}
\I(M;X_{\Out(S)}^{[n]}) = 0.
\end{align}

Equation (\ref{TH_CONV_EQ11}), in turn implies that the rate realizable by the network code $\mathcal{N}$ is $R = 0$ which is a contradiction. 

\subsubsection{$\mathbf{z \leq C_{K-T}}$}

Suppose, by contradiction, that $z > C_{K-T}$, specifically assuming that $z = C_{K-T}+1$. This implies that the eavesdropping adversary may choose to wiretap all the edges in $\mathbb{C}_{K-T}$ and any edge $e \in \In(T)$ to obtain the wiretapped set $\mathcal{W} = \mathbb{C}_{K-T} \cup \{e\}$ of size $z$. Then the wiretapped information is $X_\mathcal{W}^{[n]} = (X_{K-T}^{[n]},X_e^{[n]})$, where $X_e^{[n]}$ is the information on the chosen edge $e$. For $z$-security, we require that the mutual information $\I(M;X_\mathcal{W}^{[n]}) = 0$. Therefore,

\begin{align}
\label{TH_CONV_EQ12}
\I(M;X_\mathcal{W}^{[n]}) &= \I(M;X_{K-T}^{[n]}) + \I(M;X_e^{[n]}|X_{K-T}^{[n]})= 0.
\end{align}

Further implying that,

\begin{align}
\label{TH_CONV_EQ13}
\I(M;X_{K-T}^{[n]}) &= 0,\\
\label{TH_CONV_EQ14}
\I(M;X_e^{[n]}|X_{K-T}^{[n]}) &= 0.
\end{align}

From (\ref{TH_CONV_EQ14}), 

\begin{align}
    \label{TH_CONV_EQ15}
    \Hi(X_e^{[n]}|X_{K-T}^{[n]}) = \Hi(X_e^{[n]}|X_{K-T}^{[n]},M).
\end{align}

We now consider the cut $\mathbb{C}_{K-T}$ and the corresponding partitions $\mathcal{A}$ and $\bar{\mathcal{A}}$. Note that corresponding to the cut $\mathbb{C}_{K-T}$, the set of information and key generating source nodes in $\mathcal{G}_{\rm cyc}$ which are also present in $\bar{\mathcal{A}}$ is $\mathcal{S}_{\bar{\mathcal{A}}}$ where $\mathcal{S}_{\bar{\mathcal{A}}} \subseteq \{S,K\}$. 




Note that $\In(T) \subset \mathbb{C}_{K-T} \cup \mathcal{E}_{\bar{\mathcal{A}}}$. Due to the cut $\mathbb{C}_{K-T}$, it follows that either $S \in \mathcal{V}_{\bar{\mathcal{A}}}$ or $S \in \mathcal{V}_{\mathcal{A}}$. For any edge $e' \in \In(T) \cap \mathbb{C}_{K-T}$, we have the following.
\begin{align}
    \label{TH_CONV_EQ16}
    X_{e'}^{[n]} &= h_{e'}(X_{K-T}^{[n-1]},M), 
\end{align}

where, $h_{e'}$ is a deterministic function.

For any edge $e' \in \In(T) \cup \mathcal{E}_{\bar{\mathcal{A}}}$, by Lemma \ref{LEM_CONV}, we have the following.
\begin{align}
    \label{TH_CONV_EQ17}
    X_{e'}^{[n]} &= \{g_{e'}^{(j)}(X_{K-T}^{[j-1]},M)\}_{j \in [n]}\nonumber\\
    &= h_{e'}(X_{K-T}^{[n-1]},M).
\end{align}

Equation (\ref{TH_CONV_EQ17}) shows that for any edge $e \in \In(T)$, the random variable $X_{e}^{[n]}$ is completely determined by $X_{K-T}^{[n-1]}$ and $M$. As $\Hi(X_{e}^{[n]}|X_{K-T}^{[n]},M) \leq \Hi(X_{e}^{[n]}|X_{K-T}^{[n-1]},M)$, we have.
\begin{align}
    \label{TH_CONV_EQ18}
    \Hi(X_{e}^{[n]}|X_{K-T}^{[n]},M) = 0,
\end{align}

which implies by (\ref{TH_CONV_EQ15}) that $\Hi(X_{e}^{[n]}|X_{K-T}^{[n]}) = 0$. This holds for all $e \in \In(T)$ and thus $\Hi(X_{\In(T)}|X_{K-T}^{[n]}) = 0$. As (\ref{TH_CONV_EQ13}) shows that $X_{K-T}^{[n]}$ is independent of $M$, therefore we conclude that, 

\begin{align}
    \label{TH_CONV_EQ1A}
    \I(M;X_{\In(T)}^{[n]}) = 0.
\end{align}

This in turn implies that the rate realizable by the network code $\mathcal{N}$ is $R=0$ which is a contradiction. Thus, for $R > 0$, an $(R,z)$-feasible network code exists only if $z \leq \min(C_{K-S},C_{K-T})$, i.e., bound (\ref{BND1}) holds.

\subsection{Upper Bound of $R$}

The bound (\ref{BND2}) is a direct consequence of Theorem 2.1 of \cite{8187170} and therefore the proof is not included here.

\subsection{Upper Bound of $R+z$: $R+z \leq C_{KS-T}$}


To show that an $(R,z)$-feasible network code exists only if bound (\ref{BND3}) holds, we start by considering the following cases:
\begin{itemize}
\item \textbf{Case 1:} $z \geq C_{KS-T}$
\item \textbf{Case 2:} $z < C_{KS-T}$ 
\end{itemize}

For \textbf{Case 1}, we see that the eavesdropping adversary has the option of wiretapping all the edges in $\mathbb{C}_{KS-T}$. Therefore, we set $\mathbb{C}_{KS-T} \subseteq \mathcal{W}$ thereby forcing $\I(M;X_{KS-T}^{[n]}) = 0$. 
This, implies that $X_{KS-T}^{[n]}$ is independent of the message symbols $M$. We also observe that  $\In(T) \subset \mathbb{C}_{KS-T} \cup \mathcal{E}_{\bar{\mathcal{A}}}$. From our previous discussions, we note that the random variable $X_{\In(T)}^{[n]}$ is a deterministic function of $X_{KS-T}^{[n]}$ and therefore is also independent $M$. Thus, the terminal $T$ receives no information regarding the message symbols $M$ and therefore the rate realizable by the network code in this case is $R = 0$ which is a contradiction.

For \textbf{Case 2}, let $\mathcal{W} \subset \mathbb{C}_{KS-T}$. Then, $\mathbb{C}_{KS-T} = \mathcal{W} \cup \mathcal{W}^{C}$ where, $\mathcal{W}^C = \mathbb{C}_{KS-T} \setminus \mathcal{W}$. We denote the information on the edges of the set $\mathcal{W}^{C}$ as $X_{\mathcal{W}^C}^{[n]}$ and thus we have that $X_{KS-T}^{[n]} = (X_{\mathcal{W}}^{[n]},X_{\mathcal{W}^C}^{[n]})$ where $\Hi(X_{\mathcal{W}^C}^{[n]}) \leq n(C_{KS-T} - z)$. Here, our measure $\Hi(.)$ of entropy equals 1 for a uniform random variable in $\mathbb{F}_q$. Thus, we have the following.

\begin{align}
\label{TH_CONV_EQ20}
nR &= \I(M;X_{KS-T}^{[n]})\\
&= \I(M;X_{\mathcal{W}}^{[n]},X_{\mathcal{W}^C}^{[n]})\nonumber\\
&= \I(M;X_{\mathcal{W}}^{[n]}) + \I(M;X_{\mathcal{W}^C}^{[n]}|X_{\mathcal{W}}^{[n]})\nonumber\\
\label{TH_CONV_EQ21}
&= \I(M;X_{\mathcal{W}^C}^{[n]}|X_{\mathcal{W}}^{[n]})\\
&= \Hi(X_{\mathcal{W}^C}^{[n]}|X_{\mathcal{W}}^{[n]}) - \Hi(X_{\mathcal{W}^C}^{[n]}|X_{\mathcal{W}}^{[n]},M)\nonumber\\
&\leq \Hi(X_{\mathcal{W}^C}^{[n]}|X_{\mathcal{W}}^{[n]})\nonumber\\
&\leq \Hi(X_{\mathcal{W}^C}^{[n]})\nonumber\\
\label{TH_CONV_EQ22}
&\leq n(C_{KS-T} - z).
\end{align}

Here, (\ref{TH_CONV_EQ20}) is due to our assumption of correctly decoding $M$ and the min-cut max-flow theorem as the cut $\mathbb{C}_{KS-T}$ is an $(S-T)$-cut. (\ref{TH_CONV_EQ21}) is due to the security condition. Thus, one may realize an $(R,z)$-feasible network code over the network $\mathcal{G}_{\rm cyc}$ only if the bound (\ref{BND3}) holds for integers $R > 0$ and $z$. 

Combining our analysis for bounds (\ref{BND1}), (\ref{BND2}) and (\ref{BND3}) proves the theorem. 

\end{proof}

\section{Proof of Lemma \ref{LEM_CONV}}
\label{CS_LEM}
\begin{lemmaproof}
We prove this lemma using an induction hypothesis on the time step parameter $i$. At time $i=0$, we assume that the network edges do not carry any information. Thus, at time $i=1$, the information on all network edges in $\bar{\mathcal{A}}$ are solely a function of the random variable $X_{\bar{\mathcal{A}}}$. 

We assume by induction that the hypothesis holds for $1 < i \leq I-1$, i.e. for $i=I-1$, we have the following.
\begin{align}
  \label{CS_LEM_EQ2}
    X_{e}^{(I-1)} &= g_e^{(I-1)}(\{X_{\mathbb{C}}^{(j)}\}_{j \in [I-2]},X_{\bar{\mathcal{A}}}).
\end{align}

We now consider time step $i = I$. For an edge $e = (u,v)$, $X_e^{(N)}$ is a function of the incoming edges to $u$ and $X_{\bar{\mathcal{A}}}$ (the latter only if $u \in \mathcal{S}_{\bar{\mathcal{A}}}$). Namely,
$$ X_e^{(I)} = f_{e}^{(I)}(\{X_{e'}^{(I-1)}\}_{e' \in \In(u)},X_{\bar{\mathcal{A}}}).$$
For $u \in \mathcal{V}_{\bar{\mathcal{A}}}$, the incoming edges of $u$ are either included in the cut $\mathbb{C}$ or are in $\mathcal{E}_{\bar{\mathcal{A}}}$. Thus,
$$ X_e^{(I)} = f_{e}^{(I)}(X_\mathbb{C}^{[I-1]}, \{X_{e'}^{(I-1)}\}_{e' \in \In(u)},X_{\bar{\mathcal{A}}}).$$
By induction, as $X_{e'}^{(I-1)}$ is a function of $X_\mathbb{C}^{[I-2]}$ and $X_{\bar{\mathcal{A}}}$ for $e' \in \mathcal{E}_{\bar{\mathcal{A}}}$. We conclude that there exists a function $g_e^{(I)}$ such that,
$$ X_e^{(I)} = g_e^{(I)}(X_\mathbb{C}^{[I-1]},X_{\bar{\mathcal{A}}}).$$
This proves the lemma.
\end{lemmaproof}

\section{Proof of Claims}

\label{CLAIMPROOFS}

\subsection{Proof of Claim \ref{CL_RED}}
\label{CLAIMPROOF_CL_RED}
\begin{proof}[\unskip\nopunct]
To prove that the network code $\mathcal{N}$ is $R$-decodable over network $\mathcal{G}$, given that $\mathcal{N}^*$ is $\mathbf{R}$-feasible over $\mathcal{G}^*$, we consider the following steps.

\begin{itemize}
    \item [1.] We disconnect terminal $T^*$ from $S$ by removing the edges connecting $S$ to $T^*$.
    \item [2.] We keep the random assignment of the local coding coefficients for each edge in $\mathcal{E}$ unchanged.
    \item [3.] As in $\mathcal{N}^*$, the source $S$ does not decode the keys $N$ but ``mixes'' the incoming combinations of the keys in $N$ with the message $M$ that it holds, and transmits the resulting combinations through $\Out(S)$.
    \item [4.] The information that terminal $T$ wants to decode also remains unchanged.
\end{itemize}

By initiating the steps above, we obtain the network code $\mathcal{N}$ from $\mathcal{N}^*$. It also follows that the network code $\mathcal{N}$ allows $T$ to decode all $R$ symbols of $M$ as the $\mathbf{R}$-feasible $\mathcal{N}^*$ allows $T$ to decode all $R$ symbols of $M$ and $z$ symbols of $N$, thereby satisfying condition (\ref{NC_EQ1}). This proves the claim. 
\end{proof}

\subsection{Proof of Claim \ref{SEC_CL_1}}
\label{CLAIMPROOF_SEC_CL_1}
\begin{proof}[\unskip\nopunct]
Since $\rk(\irow{\mathbf{U}_K & \mathbf{U}_S}) = z$, we assume, without loss of generality, that the first $\sigma_K$ columns of $\mathbf{U}_K$ and $\sigma_S$ columns of $\mathbf{U}_S$ are jointly linearly independent with $\sigma_K + \sigma_S = z$. Then, we have that $\mathbf{U}_S = \irow{\mathbf{\Bar{U}}_S & \mathbf{\Bar{U}}_S\mathbf{\Gamma}}$, where $\mathbf{\bar{U}}_S$ is $z \times \sigma_S$ matrix of full column-rank $\sigma_S$, and $\mathbf{\Gamma}$ is a $\sigma_S \times (O_S - \sigma_S)$ matrix. Let $\mathbf{\Bar{U}}_K$ be the sub-matrix of $\mathbf{U}_K$ containing the $\sigma_K$ linearly independent columns of $\mathbf{U}_K$, and $\mathbf{\Delta}$ be a $\sigma_K \times O_K$ matrix such that $\mathbf{U}_K = \mathbf{\Bar{U}}_K\mathbf{\Delta}$. Then, as per our assumption, the first $\sigma_K$ columns $\mathbf{\Delta}$ form a $\sigma_K \times \sigma_K$ identity matrix. Therefore, we have that $\rk(\mathbf{\Delta}) = \sigma_K$

We now consider the matrix $\mathbf{V}_{\In(S)}$. Since, $\rk(\mathbf{V}_{\In(S)}) = z$, we may express $\mathbf{V}_{\In(S)}$ as follows.

\begin{align}
    \label{SEC_CL_EQ_a}
    \mathbf{V}_{\In(S)} &= \begin{bmatrix}
    \mathbf{V} & \mathbf{V}_1\\
    \mathbf{V}_2 & \mathbf{V}_3
    \end{bmatrix}
\end{align}

Here, $\mathbf{V}$ is a $z \times z$ invertible sub-matrix of $\mathbf{V}_{\In(S)}$. The columns of the $z \times (O_K-z)$ sub-matrix $\mathbf{V}_1$ and the rows of the $(I_S-z) \times z$ sub-matrix $\mathbf{V}_2$ are spanned by the columns and rows of $\mathbf{V}$ respectively, while the rows of the sub-matrix $\mathbf{V}_3$ are spanned by the rows of $\mathbf{V}_1$. Thus we may rewrite \eqref{SEC_CL_EQ_a} as follows

\begin{align}
    \label{SEC_CL_EQ_a_1}
    \mathbf{V}_{\In(S)} &= \begin{bmatrix}
    \mathbf{V} & \mathbf{V}_1\\
    \mathbf{A}\mathbf{V} & \mathbf{V}_3
    \end{bmatrix}
\end{align}

Here, $\mathbf{A}$ is  an $(I_S - z) \times z$ matrix. We now partition the matrix $\mathbf{B}_S  = \irow{\mathbf{B}_{S,1} & \mathbf{B}_{S,2}}$ such that $\mathbf{B}_{S,1}$ and $\mathbf{B}_{S,2}$ are $O_S \times z$ and $O_S \times (I_S-z)$ matrices respectively. Let $\mathbf{B} \triangleq \mathbf{B}_S\mathbf{V}_{\In(S)}$, then by using \eqref{SEC_CL_EQ_a_1} we obtain the following.

\begin{align}
    \label{SEC_CL_EQ_b}
    \mathbf{B} &= \irow{\mathbf{B}_{S,1}\mathbf{V}+\mathbf{B}_{S,2}\mathbf{A}\mathbf{V} & \mathbf{B}_{S,1}\mathbf{V}_1+\mathbf{B}_{S,2}\mathbf{V}_3}\nonumber\\
    &\triangleq\irow{\mathbf{\bar{B}} & \mathbf{\Tilde{B}}}
\end{align}

Here, the matrices $\mathbf{\bar{B}}$ and $\mathbf{\Tilde{B}}$ are $O_S \times z$ and $O_S \times (I_S-z)$, respectively. Furthermore, we partition $\mathbf{\bar{B}} = \icol{\mathbf{\bar{B}}^1 \\ \mathbf{\bar{B}}^2}$, where $\mathbf{\bar{B}}^1$ and $\mathbf{\bar{B}}^2$ are $\sigma_S \times z$ and $(O_S-\sigma_S) \times z$ sub-matrices respectively. Likewise, we partition $\mathbf{\Tilde{B}} = \icol{\mathbf{\Tilde{B}}^1 \\ \mathbf{\Tilde{B}}^2}$. Then, using \eqref{SEC_CL_EQ_a_1} and \eqref{SEC_CL_EQ_b}, we may rewrite \eqref{SEC_EQ6} as follows.

\begin{align}
    \label{SEC_CL1_EQ_c}
    \mathbf{\Phi} &= \irow{\mathbf{\bar{U}}_K & \mathbf{\bar{U}}_S}\begin{bmatrix}
    \mathbf{\Delta}\\ \\
    \mathbf{\bar{B}}^1+\mathbf{\Gamma}\mathbf{\bar{B}}^2 & \mathbf{\Tilde{B}}^1+\mathbf{\Gamma}\mathbf{\Tilde{B}}^2
    \end{bmatrix}
\end{align}

We now consider partition $\mathbf{\Delta} = \irow{\mathbf{\Delta}_z & \mathbf{\Hat{\Delta}}}$, where $\mathbf{\Delta}_z$ consists of the first $z$ columns of $\mathbf{\Delta}$. Then, we have that $\mathbf{\Delta}_z = \irow{\mathbf{I}_{\sigma_K} & \mathbf{\Bar{\Delta}}_z}$ where $\mathbf{I}_K$ is the $\sigma_K \times \sigma_K$ identity matrix and $\mathbf{\Bar{\Delta}}_z$ is a $\sigma_K \times (z-\sigma_K)$ matrix whose columns are spanned by $\mathbf{I}_{\sigma_K}$. Thus, we see that the $\sigma_K$ rows of $\mathbf{\Delta}_z$ are linearly independent. We may rewrite \eqref{SEC_CL1_EQ_c} as follows.

\begin{align}
    \label{SEC_CL1_EQ_d}
    \mathbf{\Phi} &= \irow{\mathbf{\bar{U}}_K & \mathbf{\bar{U}}_S}\begin{bmatrix}
    \mathbf{\Delta}_z & \mathbf{\Hat{\Delta}}\\ \\
    \mathbf{\bar{B}}^1+\mathbf{\Gamma}\mathbf{\bar{B}}^2 & \mathbf{\Tilde{B}}^1+\mathbf{\Gamma}\mathbf{\Tilde{B}}^2
    \end{bmatrix}\nonumber\\
    &\triangleq \irow{\mathbf{\bar{U}}_K & \mathbf{\bar{U}}_S}\mathbf{Q}
\end{align}

Since $\irow{\mathbf{\bar{U}}_K & \mathbf{\bar{U}}_S}$ is invertible, $\rk(\mathbf{Q}) = z$ implies $\rk(\mathbf{\Phi}) = z$. To prove that $\rk(\mathbf{Q}) = z$, we show that the $\sigma_S$ rows of $\mathbf{\bar{B}}^1+\mathbf{\Gamma}\mathbf{\bar{B}}^2$, having dimension $z$, are linearly independent and not spanned by the $\sigma_K$ rows of $\mathbf{\Delta}_z$. 

Given that the entries of the matrix $\mathbf{B}_S$ are i.i.d and uniform in $\mathbb{F}_q$, for any $\psi \in \mathbb{F}_q^z$ and for $i \in [O_S]$, we compute the probability $\Pr_{\mathbf{B}_S}\{(\bar{b})^i = \psi\}$, where $(\bar{b})^i$ denotes the $i$-th row of $\mathbf{\bar{B}}$ of dimension $z$. From \eqref{SEC_CL_EQ_b}, we see that 
\begin{align}
    \label{SEC_CL1_EQ_e1a}
    (\Bar{b}^i) = (b_{S,1})^i\mathbf{V}+ (\mathbf{B}_{S,2}\mathbf{A})^i\mathbf{V}.
\end{align}

Here, $(b_{S,1})^i$ and $(\mathbf{B}_{S,2}\mathbf{A})^i$ denotes the $i$-th row of $\mathbf{B}_{S,1}$ and $\mathbf{B}_{S,2}\mathbf{A}$ respectively. As $\mathbf{V}$ is invertible, we have the following.

\begin{align}
    \Pr_{\mathbf{B}_S}\{(\bar{b})^i = \psi\} &= \Pr_{(b_{S,1})^i,\mathbf{B}_{S,2}}\{(b_{S,1})^i+ (\mathbf{B}_{S,2}\mathbf{A})^i = \psi\mathbf{V}^{-1}\}\nonumber\\
    &= \sum_{\bar{\psi}}\Pr_{\mathbf{B}_{S,2}}\{(\mathbf{B}_{S,2}\mathbf{A})^i = \bar{\psi}\}.\nonumber\\
    &\quad\quad \Pr_{(b_{S,1})^i}\{(b_{S,1})^i = (\psi\mathbf{V}^{-1} - \bar{\psi})\}\nonumber\\
    \label{SEC_CL1_EQ_e1}
    &= \dfrac{1}{q^z}\sum_{\bar{\psi}}\Pr_{\mathbf{B}_{S,2}}\{(\mathbf{B}_{S,2}\mathbf{A})^i = \bar{\psi}\}\\
    \label{SEC_CL1_EQ_e2}
    &= \dfrac{1}{q^z}
\end{align}

Here, \eqref{SEC_CL1_EQ_e1} is due to that fact that the entries of $\mathbf{B}_{S,1}$, which is a sub-matrix of $\mathbf{B}_S$, are i.i.d. and uniform in $\mathbb{F}_q$. From \eqref{SEC_CL1_EQ_e2}, we see that the  rows of $\mathbf{\bar{B}}$ are uniform in $\mathbb{F}_q^z$. For a fixed matrix $\mathbf{B}_{S,2}$ in \eqref{SEC_CL1_EQ_e1a} and due to the invertibility of matrix $\mathbf{V}$, there exists a 1-1 map between $(b_{S,1})^i \in \mathbb{F}_q^z$ and $(\Bar{b}^i) \in \mathbb{F}_q^z$, for all $i \in [O_S]$. Now, as the vectors $(b_{S,1})^i$ are chosen independently for each $i \in [O_S]$, it follows that the corresponding vectors $(\Bar{b}^i)$ must also be independent for all $i \in [O_S]$.
This implies that the rows of $\mathbf{\bar{B}}^1$, which is a sub-matrix of $\mathbf{\bar{B}}$ containing its first $\sigma_S$ rows, are also i.i.d and uniform in $\mathbb{F}_q^z$. 

Let $\mathbf{C} \triangleq \mathbf{\bar{B}}^1+\mathbf{\Gamma}\mathbf{\bar{B}}^2$ and $c^i$ denote the $i$-th row of $\mathbf{C}$ for $i \in [\sigma_S]$. For any $\rho \in \mathbb{F}_q^z$, we compute the probability $\Pr_{\mathbf{\bar{B}}}\{c^i = \rho\}$. Denoting the $i$-th rows of $\mathbf{\bar{B}}^1$ and $\mathbf{\Gamma}\mathbf{\bar{B}}^2$ as $(\bar{b}^1)^i$ and $(\mathbf{\Gamma}\mathbf{\bar{B}}^2)^i$, respectively, we have that, $$c^i = (\bar{b}^1)^i + (\mathbf{\Gamma}\mathbf{\bar{B}}^2)^i.$$ Note that the rows of $\mathbf{\bar{B}}^1$ form the first $\sigma_S$ rows of $\mathbf{\bar{B}}$ and therefore are i.i.d. and uniform in $\mathbb{F}_q^z$. Thus, by applying the same argument as in \eqref{SEC_CL1_EQ_e2}, we obtain $\Pr_{\mathbf{\bar{B}}}\{c^i = \rho\} = \dfrac{1}{q^z}$. The vectors $\{c^i\}_{i \in [O_S]}$ are mutually independent due to the fact that the vectors $\{(\bar{b}^1)^i\}_{i \in [O_S]}$ are mutually independent. Thus, as $\sigma_K+\sigma_S = z$, we have the following.




\begin{align}
    \label{SEC_CL1_EQ_g}
    \Pr_{\mathbf{B}_S}\{\rk(\Phi) = z\} &= \dfrac{\prod_{l=0}^{\sigma_S-1}q^{z} - q^{\sigma_K+l}}{q^{z\sigma_S}}\nonumber\\
    &= \prod_{l=0}^{\sigma_S-1}\Big(1 - \dfrac{1}{q^{z-\sigma_k - l}}\Big)\nonumber\\
    &> \Big(1 - \dfrac{1}{q}\Big)^{\sigma_S}\nonumber\\
    &> 1 - \dfrac{\sigma_S}{q}\nonumber\\
    &> 1 - \dfrac{z}{q}
\end{align}

This proves our claim.

\end{proof}
\subsection{Proof of Claim \ref{SEC_CL_2}}
\label{CLAIMPROOF_SEC_CL_2}
\begin{proof}[\unskip\nopunct]

Let $\mathbf{AB} = \irow{\lambda_1 & \lambda_2 & \cdots & \lambda_n}$, where $\lambda_j \in \mathbb{F}_q^n$ for $j \in [n]$, $\mathbf{A} = \irow{a_1 & a_2 & \cdots & a_m}$, where $a_i \in \mathbb{F}_q^n$ for $i \in [m]$ and $\mathbf{B} = \{b_{i,j}\}_{i \in [m], j \in [n]}$, where $b_{i,j}$'s are i.i.d. and uniform over $\mathbb{F}_q$.

For any $j \in [n]$, we have
\begin{align}
    \label{SEC_CL2_EQ1}
    \lambda_j &= \sum_{i \in [m]}a_i b_{i,j}
\end{align}

For any $\omega \in \mathbb{F}_q^n$, we first compute  $\Pr_{b_j}\{\lambda_j = \omega\}$, where $b_j$ is the $m$-dimensional $j$-th column of $\mathbf{B}$. Since $\rk(\mathbf{A}) = n$, we assume, without loss of generality, that the last $n$ columns of $\mathbf{A}$ are linearly independent. We also partition $b_j$ such that $b_j = \icol{\Bar{b}_j \\ \Tilde{b}_j}$ where $\Tilde{b}_j$ consists of the last $n$ entries of $b_j$. Then, we may rewrite (\ref{SEC_CL2_EQ1}) as

\begin{align}
    \label{SEC_CL2_EQ2}
    \lambda_j &= \sum_{i \in [m-n]}a_i b_{i,j} + \sum_{i = m-n+1}^{m}a_i b_{i,j}
\end{align}
Then,
\begin{align}
    \label{SEC_CL2_EQ3}
    \Pr_{b_j}\{\lambda_j = \omega\} &= \Pr_{b_j}\{\sum_{i \in [m]}a_i b_{i,j} = \omega\}\nonumber\\
    &= \sum_{\Hat{\omega}}\Pr_{\bar{b}_j}\{\sum_{i \in [m-n]}a_i b_{i,j} = \Hat{\omega}\}.\nonumber\\
    &\quad\quad \Pr_{\Tilde{b}_j}\{\sum_{i = m-n+1}^{m}a_i b_{i,j} = \omega - \Hat{\omega}\}
\end{align}

Since the $n$-dimensional columns $\{a_i\}_{i = m-n+1}^{m}$ are linearly independent, the $n$-system of equations $\sum_{i = m-n+1}^{m}a_i b_{i,j} = \omega - \Hat{\omega}$ must have a unique solution for each $\Tilde{b}_j \in \mathbb{F}_q^n$, and as the entries of $\Tilde{b}_j$ are i.i.d. uniform in $\mathbb{F}_q$, we have that $\Pr_{\Tilde{b}_j}\{\sum_{i = m-n+1}^{m}a_i b_{i,j} = \omega - \Hat{\omega}\} = 1/q^n$. Thus, we may rewrite (\ref{SEC_CL2_EQ3}) as follows.

\begin{align}
    \label{SEC_CL2_EQ4}
    \Pr_{b_j}\{\lambda_j = \omega\} &= \dfrac{1}{q^n}\sum_{\Hat{\omega}}\Pr_{\bar{b}_j}\{\sum_{i \in [m-n]}a_i b_{i,j} = \Hat{\omega}\}\nonumber\\
    &= \dfrac{1}{q^n}
\end{align}

Equation (\ref{SEC_CL2_EQ4}) implies that the columns $\{\lambda_j\}_{j \in [n]}$ are uniform in $\mathbb{F}_q^n$. The columns $\lambda_j$ are also mutually independent due to the fact that the columns $b_j$ are independent for all $j \in [n]$. Thus, we have

\begin{align}
    \label{SEC_CL2_EQ5}
    \Pr_{\mathbf{B}}\{\rk(\mathbf{A}\mathbf{B}) = n\} &= \dfrac{\prod_{l = 0}^{n-1}\Big(q^n - q^l\Big)}{q^{n^2}}\nonumber\\
    &= \prod_{l = 0}^{n-1}\Big(1 - \dfrac{1}{q^{n-l}}\Big)\nonumber\\
    &> \Big(1 - \dfrac{1}{q}\Big)^n\nonumber\\
    &> 1 - \dfrac{n}{q}
\end{align}

This proves the claim.

\end{proof}

\end{appendices}

\end{document}